\documentclass[11pt]{amsart}
\usepackage[a4paper,margin=1in]{geometry}

\usepackage{amsmath}
\usepackage{amsthm}
\usepackage{mathtools} 
\usepackage[algoruled]{algorithm2e}
\usepackage{enumitem}
\usepackage[T1]{fontenc}
\usepackage{tikz}
\usetikzlibrary{arrows,arrows.meta,backgrounds,calc,fit,decorations.pathreplacing,decorations.markings,shapes.geometric}

\tikzstyle{internal} = [draw, fill, shape=circle]
\tikzstyle{external} = [shape=circle]
\tikzstyle{square}   = [draw, fill, rectangle]
\tikzstyle{triangle} = [draw, fill, regular polygon, regular polygon sides=3, inner sep=3pt]
\tikzstyle{pentagon} = [draw, fill, regular polygon, regular polygon sides=5, inner sep=2pt, minimum size=14pt]

\usepackage[margin=1cm]{caption} 
\usepackage{subfig}

\usepackage[thinlines]{easytable}

\tikzset{every fit/.append style=text badly centered}

\usepackage[bookmarks=true,hypertexnames=false,pagebackref]{hyperref}
\hypersetup{colorlinks=true, citecolor=blue, linkcolor=red, urlcolor=blue}

\usepackage{scrtime}
\usepackage{ifthen}

\usepackage{cleveref}

\usepackage{todonotes}
\usepackage{mleftright}
\usetikzlibrary{positioning,chains,fit,shapes,calc}
\usetikzlibrary{trees}
\usetikzlibrary{decorations.pathreplacing}
\usetikzlibrary{decorations.pathmorphing}
\usetikzlibrary{decorations.markings}
\tikzset{>=latex} 

\usepackage{cool}
\Style{DSymb={\mathrm d},DShorten=true,IntegrateDifferentialDSymb=\mathrm{d}}

\newcommand{\numP}{\#{\bf P}}
\newcommand{\NP}{{\bf NP}}
\newcommand{\vbl}{{\sf var}}
\newcommand{\arc}{{\sf arc}}

\newcommand{\prs}{partial rejection sampling} 

\newcommand{\aux}{{\sf Aux}}

\newcommand{\Ex}{\mathop{\mathbb{{}E}}\nolimits}

\def\*#1{\mathbf{#1}}
\def\+#1{\mathcal{#1}}
\def\-#1{\mathrm{#1}}

\SetKwProg{Fn}{Function}{:}{end}
\SetKwInOut{Input}{Input}
\SetKwInOut{Output}{Output}

\newcommand{\abs}[1]{\left\vert#1\right\vert}
\newcommand{\ceil}[1]{\lceil#1\rceil}
\newcommand{\floor}[1]{\lfloor#1\rfloor}
\newcommand{\set}[1]{\left\{#1\right\}}
\newcommand{\eps}{\varepsilon}

\renewcommand{\Pr}{\mathop{\mathrm{Pr}}\nolimits}

\newcommand{\Zrel}{Z_{\textnormal{rel}}}

\newcommand{\reverse}[1]{\ensuremath{\overline{#1}}}
\newcommand{\sfix}{\ensuremath{S_{\fix}}}
\newcommand{\sigmafix}{\ensuremath{\sigma_{\fix}}}

\newcommand{\wt}{\mathrm{wt}}
\newcommand{\fix}{\mathrm{fix}}

\newcommand{\maxx}{\mathrm{max}}
\newcommand{\minn}{\mathrm{min}}

\newtheorem{theorem}{Theorem}

\newtheorem{lemma}[theorem]{Lemma}
\newtheorem{claim}[theorem]{Claim}

\newtheorem{proposition}[theorem]{Proposition}
\newtheorem{corollary}[theorem]{Corollary}
\newtheorem{condition}[theorem]{Condition}
\newtheorem{definition}[theorem]{Definition}
\newtheorem*{remark}{Remark}

\crefname{theorem}{Theorem}{Theorems}
\crefname{observation}{Observation}{Observations}
\crefname{claim}{Claim}{Claims}
\crefname{condition}{Condition}{Conditions}
\crefname{algorithm}{Algorithm}{Algorithms}
\crefname{property}{Property}{Properties}
\crefname{example}{Example}{Examples}
\crefname{fact}{Fact}{Facts}
\crefname{lemma}{Lemma}{Lemmas}
\crefname{corollary}{Corollary}{Corollaries}
\crefname{definition}{Definition}{Definitions}
\crefname{remark}{Remark}{Remarks}
\crefname{proposition}{Proposition}{Propositions}
\crefname{equation}{equation}{equations}

\makeatletter
\def\prob#1#2#3{\goodbreak\begin{list}{}{\labelwidth\z@ \itemindent-\leftmargin
                        \itemsep\z@  \topsep6\p@\@plus6\p@
                        \let\makelabel\descriptionlabel}
                      \item[\textbf{Name}]#1
                      \item[\textbf{Instance}]#2
                      \item[\textbf{Output}]#3
                \end{list}}
\makeatother

\title{Tight bounds for popping algorithms}

\author[H.\ Guo]{Heng Guo}
\address[Heng Guo]{School of Informatics, University of Edinburgh, Informatics Forum, Edinburgh, EH8 9AB, United Kingdom.}
\email{hguo@inf.ed.ac.uk}

\author[K.\ He]{Kun He}
\address[Kun He]{Institute of Computing Technology, Chinese Academy of Sciences, No.6 Kexueyuan South Road, 100190, Beijing, China}
\email{hekun@ict.ac.cn}

\begin{document}

\begin{abstract}
  We sharpen run-time analysis for algorithms under the partial rejection sampling framework.
  Our method yields improved bounds for
  \begin{itemize}
    \item the cluster-popping algorithm for approximating all-terminal network reliability;
    \item the cycle-popping algorithm for sampling rooted spanning trees;
    \item the sink-popping algorithm for sampling sink-free orientations.
  \end{itemize}
  In all three applications, our bounds are not only tight in order, but also optimal in constants.
\end{abstract}

\maketitle

\section{Introduction}

The counting complexity class \numP\ was introduced by Valiant \cite{Val79b} to capture the apparent intractability of the permanent.
Although exactly evaluating \numP-complete problems is a task even harder than solving \NP-complete problems~\cite{Tod91},
efficient approximation algorithms may still exist.
This possibility was first exhibited through the relationship between approximate counting and sampling \cite{JVV86},
and in particular, using the Markov chain Monte Carlo (MCMC) method \cite{Jer03}.
Early successes along this line include efficient algorithms to approximate the volume of a convex body \cite{DFK91},
to approximate the partition function of ferromagnetic Ising models \cite{JS93},
and to approximate the permanent of a non-negative matrix \cite{JSV04}.
More recently, a number of exciting new approaches were developed,
utilising a variety of tools such as correlation decay \cite{Wei06,BG08}, zeros of graph polynomials \cite{Bar16,PR17}, and Lov\'asz local lemma \cite{Moi17}.

Partial rejection sampling \cite{GJL17} is yet another alternative approach to approximate counting and exact sampling.
It takes an algorithmic Lov\'asz local lemma \cite{MT10} perspective for Wilson's cycle-popping algorithm \cite{PW98} to sample rooted spanning trees.
The algorithm is extremely simple.
In order to sample from a product distribution conditional on avoiding a number of undesirable ``bad'' events,
we randomly initialise all variables,
and re-sample a subset of variables selected by a certain subroutine, until no bad event is present.
In the so-called extremal instances, the resampled variables are just those involved in occurring bad events.
Despite its simplicity, it can be applied to a number of situations that are seemingly unrelated to the local lemma.
Unlike Markov chains, partial rejection sampling yields exact samplers.
The most notable application is the first polynomial-time approximation algorithm for all-terminal network reliability \cite{GJ18a}.

In this paper we sharpen the run-time analysis for a number of algorithms under the partial rejection sampling framework for extremal instances.
We apply our method to analyse cluster-, cycle-, and sink-popping algorithms.
Denote by $n$ the number of vertices in a graph,
and $m$ the number of edges (or arcs) in a undirected (or directed) graph.
We summarise some background and our results below.
\begin{itemize}
  \item
    Cluster-popping is first proposed by Gorodezky and Pak \cite{GP14} to approximate a network reliability measure, called reachability, in directed graphs.
    They conjecture that the algorithm runs in polynomial-time in expectation on bi-directed graphs,
    which is confirmed by Guo and Jerrum \cite{GJ18a}.
    It has also been shown in \cite{GJ18a} that a polynomial-time approximation algorithm for all-terminal network reliability can be obtained using cluster-popping.

    We show a $\frac{p_\maxx}{1-p_\maxx}mn$ upper bound for the expected number of resampled variables on bi-directed graphs,
    where $p_\maxx$ is the maximum failure probability on edges.
    This improves the previous $\frac{p_\maxx}{1-p_\maxx}m^2n$ bound \cite{GJ18a}.
    We also provide an efficient implementation based on Tarjan's algorithm \cite{Tar72} to obtain a faster algorithm to sample spanning connected subgraphs in $O(mn)$ time  (assuming that $p_{\maxx}$ is a constant).
    Furthermore, we obtain a faster approximation algorithm for the all-terminal network reliability,
    which takes $O(mn^2\log n)$ time, improving upon $O(m^2n^3)$ from \cite{GJ18a}.
    For precise statements, see \Cref{thm:cluster-time} and \Cref{thm:reliability-faster}.
  \item
    Cycle-popping is introduced by Propp and Wilson \cite{Wilson96,PW98} to sample uniformly rooted spanning trees, a problem with a long line of history.
    We obtain a $2mn$ upper bound for the expected number of resampled variables,
    improving the constant from the previous $O(mn)$ upper bound \cite{PW98}.
    See \Cref{thm:cycle-time}.
  \item
    Sampling sink-free orientations is introduced by Bubley and Dyer \cite{BD97a}
    to show that the number of solutions to a special class of CNF formulas is \numP-hard to count exactly,
    but can be counted in polynomial-time approximately.
    Bubley and Dyer showed that the natural Markov chain takes $O(m^3\log 1/\eps)$ time
    to generate a distribution $\eps$-close (in total variation distance) to the uniform distribution over all sink-free orientations.
    Using the coupling from the past technique, Huber \cite{Hub98} obtained a $O(m^4)$ exact sampler.
    Cohn, Pemantle, and Propp \cite{CPP02} introduced an alternative exact sampler, called sink-popping.
    They show that sink-popping resamples $O(mn)$ random variables in expectation.

    We improve the expected number of resampled variables for sink-popping to at most $n(n-1)$.
    See \Cref{thm:sink-time}.
\end{itemize}
In all three applications, we also construct examples to show that none of the constants (if explicitly stated) can be further improved.
Our results yield best known running time for all problems mentioned above except sampling uniform spanning trees,
for which the current best algorithm by Schild \cite{Sch18} is in almost-linear time in $m$.
We refer interested readers to the references therein for the vast literature on this problem.
One should note that the corresponding counting problem for spanning trees is tractable via Kirchhoff's matrix-tree theorem,
whereas for the other two exact counting is \numP-hard \cite{Jer81,PB83,BD97a}.
It implies that their solution spaces are less structured, and renders sampling much more difficult than for spanning trees.

The starting point of our method is an exact formula \cite[Theorem 13]{GJL17} for the expected number of events resampled,
for partial rejection sampling algorithm on extremal instances.
Informally, it states that the expected number of resampled events equals to
the ratio between the probability (under the product distribution) that exactly one bad event happens, and the probability that none of the bad events happens.
This characterisation has played an important role in the confirmation of the conjecture by Gorodezky and Pak \cite{GP14},
which leads to the aforementioned all-terminal network reliability algorithm \cite{GJ18a}.

When bad events involve only constant number of variables,
bounding this ratio is sufficient to obtain tight run-time bound.
However, in all three popping algorithms mentioned above, bad events can involve as many variables as $m$ or $n$,
and simply applying a worst case bound (such as $m$ or $n$) yields loose run-time upper bound.
Our improvement comes from a refined expression of the number of variables, rather than events, resampled in expectation.
(See \eqref{eqn:expected} and \Cref{thm:run-time}.)
We then apply a combinatorial encoding idea, and design injective mappings to bound the refined expression.
Similar mappings have been obtained before in \cite{GJL17,GJ18a},
but our new mappings are more complicated and carefully designed in order to achieve tight bounds.
We note that our analysis is completely different and significantly simpler than the original analysis for cycle-popping \cite{PW98} and sink-popping \cite{CPP02}.

Since our bounds are tight, one has to go beyond partial rejection sampling to further accelerate cluster- and sink-popping.
It remains an interesting open question whether $O(mn)$ is a barrier to uniformly sample spanning connected subgraphs.

\section{Partial rejection sampling}\label{sec:PRS}

Let $\{X_1,\dots,X_n\}$ be a set of mutually independent random variables.
Each $X_i$ can have its own distribution and range.
Let $\{A_1,\dots,A_m\}$ be a set of ``bad'' events that depend on $X_i$'s.
For example, for a constraint satisfaction problem (CSP) with variables $X_i$ ($i\in[n]$) and constraints $C_j$ ($j\in[m]$),
each $A_j$ is the set of unsatisfying assignments of $C_j$ for $j\in[m]$.
Let $\vbl(A_j)$ be the set of variables that $A_j$ depends on.

Our goal is to sample from the product distribution of $(X_i)_{i\in[n]}$,
conditional on none of the bad events $(A_j)_{j\in[m]}$ occurring.
Denote by $\mu(\cdot)$ the product distribution and by $\pi(\cdot)$ the conditional one (which is our target distribution).

A breakthrough result in algorithmic Lov\'asz local lemma is the Moser and Tardos algorithm \cite{MT10},
which simply iteratively eliminates occurring bad events.
The form we will use is described in \Cref{alg:prs}.

\begin{algorithm}
  \caption{Partial rejection sampling for extremal instances}
  \label{alg:prs}
  Draw independent samples of all variables $X_1,\dots,X_n$ from their respective distributions\;
  \While{at least one bad event occurs}
  {Find the set $I$ of all occurring $A_i$\;
  Independently resample all variables in $\bigcup_{i\in I}\vbl(A_i)$\;
  }
  \KwRet{the final assignment}
\end{algorithm}

The resampling table is a very useful concept of analysing \Cref{alg:prs}, introduced by \cite{MT10},
and similar ideas were also used in the analysis of cycle-popping \cite{PW98} and sink-popping \cite{CPP02}. 
Instead of drawing random samples, 
we associate an infinite stack to each random variable.
Construct an infinite table such that each row represents random variables, 
and each entry is a sample of the random variable.
The execution of \Cref{alg:prs} can be thought of (but not really implemented this way) as first drawing the whole resampling table,
and whenever we need to sample a variable, we simply use the next value of the table.
Once the resampling table is fixed, \Cref{alg:prs} becomes completely deterministic.

In general \Cref{alg:prs} does not necessarily produce the desired distribution $\pi(\cdot)$.
However, it turns out that it does for extremal instances (in the sense of Shearer \cite{Shearer85}).

\begin{condition}  \label{cond:extremal}
  We say a collection of bad events $(A_i)_{i\in[m]}$ are \emph{extremal},
  if for any $i\neq j$, either $\vbl(A_i)\cap\vbl(A_j)=\emptyset$ or $\Pr_{\mu}(A_i\wedge A_j)=0$.
\end{condition}

In other words, if the collection is extremal,
then any two bad events are either disjoint or depend on distinct sets of variables (and therefore independent).
It was shown \cite[Theorem 8]{GJL17} that if \Cref{cond:extremal} holds,
then \Cref{alg:prs} indeed draws from $\pi(\cdot)$ (conditioned on halting).
In particular, \Cref{cond:extremal} guarantees that the set of occurring bad events have disjoint sets of variables.

\Cref{cond:extremal} may seem rather restricted,
but it turns out that many natural problems have an extremal formulation.
Examples include the cycle-popping algorithm for uniform spanning trees \cite{Wilson96,PW98},
the sink-popping algorithm for sink-free orientations \cite{CPP02},
and the cluster-popping algorithm for root-connected subgraphs \cite{GP14,GJ18a}.
In particular,
the cluster-popping algorithm yields the first polynomial-time approximation algorithm for the all-terminal network reliability problem.
More recently, another application is found to approximately count the number of bases in a bicircular matroid \cite{GJ18c},
which is known to be \numP-hard to count exactly \cite{GN06}.
We refer interested readers to \cite{GJL17,GJ18b} for further applications of the \prs\ framework beyond extremal instances.

A particularly nice feature of \prs\ algorithms for extremal instances is that we have an exact formula \cite[Theorem 13]{GJL17} for the expected number of events resampled.
This formula makes the analysis of these algorithms much more tractable.
However, it counts the number of resampled events.
In the most interesting applications,
bad events typically involve more than constantly many variables.
Using the worst case bound of the number of variables involved in bad events yields loose bounds.

We give a sharper formula for the expected number of \emph{variables} resampled next.
Let $T_i$ be the number of resamplings of event $A_i$.
Let $q_i$ be the probability such that exactly $A_i$ occurs,
and $q_{\emptyset}$ be the probability such that none of $(A_i)_{i\in [m]}$ occurs.
Suppose $q_{\emptyset}>0$ as otherwise the support of $\pi(\cdot)$ is empty.
For extremal instances, \cite[Lemma 12]{GJL17} and the first part of the proof of \cite[Theorem 13]{GJL17} yield
\begin{align}\label{eqn:expected-i}
  \Ex T_i=\frac{q_i}{q_{\emptyset}}.
\end{align}
The proof of \eqref{eqn:expected-i} is via manipulating the moment-generating function of $T_i$.
Let $T$ be the number of resampled variables.\footnote{This notation is different from that in \cite{GJL17}.}
By linearity of expectation and \eqref{eqn:expected-i},
\begin{align}\label{eqn:expected}
  \Ex T=\sum_{i=1}^m \frac{q_i\cdot\abs{\vbl(A_i)}}{q_{\emptyset}}.
\end{align}
We note that an upper bound similar to the right hand side of \eqref{eqn:expected} was first shown by Kolipaka and Szegedy \cite{KS11},
in a much more general setting but counting the number of resampled events.
We will use \eqref{eqn:expected} to derive both upper and lower bounds.

In order to upper bound the ratio in \eqref{eqn:expected},
we will use a combinatorial encoding idea, namely to design an injective mapping.
For an assignment $\sigma$, let $\wt(\sigma)$ be its weight so that $\wt(\sigma)\propto \Pr_{\mu}(\sigma)$.
Let $\Omega_{A_i}$ be the set of assignments so that exactly $A_i$ occurs,
and $\Omega_1:=\bigcup_{i=1}^m\Omega_{A_i}$.
Note that $(\Omega_{A_i})_{i\in m}$ are mutually exclusive and are in fact a partition of $\Omega_1$.
Also, let
\begin{align*}
  \Omega_1^{\vbl}:=\{(\sigma,X)\mid \exists i,\; \sigma\in\Omega_{A_i}\text{ and }X\in\vbl(A_i)\}.
\end{align*}
Moreover, let $\Omega_0$ be the set of ``perfect'' assignments such that none of $(A_i)_{i\in [m]}$ occurs.

\begin{definition}\label{def:injective-mapping}
  For a constant $r>0$ and an auxiliary set $\aux$,
  a mapping $\tau:\Omega_1^{\vbl}\rightarrow\Omega_0\times \aux$ is \emph{$r$-preserving} if
  for any $i$, any $\sigma\in \Omega_{A_i}$ and $X\in\vbl(A_i)$,
  \begin{align*}
    \wt(\sigma)\le r\cdot \wt(\tau'(\sigma)),
  \end{align*}
  where $\tau'$ is the restriction of $\tau$ on the first coordinate.
\end{definition}

A straightforward consequence of \eqref{eqn:expected} is the following theorem,
which will be our main technical tool.

\begin{theorem}\label{thm:run-time}
  If there exists a $r$-preserving injective mapping $\tau:\Omega_1^{\vbl}\rightarrow\Omega_0\times \aux$ for some auxiliary set $\aux$,
  then the expected number of resampled variables of \Cref{alg:prs} is at most $r\abs{\aux}$.
\end{theorem}
\begin{proof}
  Note that $q_i \propto \sum_{\sigma\in \Omega_{A_i}}\wt(\sigma)$
  and $q_{\emptyset} \propto \sum_{\sigma\in \Omega_{0}}\wt(\sigma)$.
  As $\tau$ is $r$-preserving and injective,
  \begin{align*}
    \sum_{i=1}^m \frac{q_i\cdot\abs{\vbl(A_i)}}{q_{\emptyset}} \le r\abs{\aux}.
  \end{align*}
  The theorem then follows from \eqref{eqn:expected}.
\end{proof}

\section{Cluster-popping}

Let $G=(V,A)$ be a directed graph with root $r$.
The graph $G$ is called \emph{root-connected} if there is a directed path in $G$ from every non-root vertex to~$r$.
Let $0<p_a<1$ be the failure probability for arc $a$,
and define the weight of a subgraph $S\subseteq A$ to be $\wt(S):=\prod_{a\in S}(1-p_a)\prod_{a\not\in S}p_a$.
Thus the target distribution $\pi(\cdot)$ is over all root-connected subgraphs,
and $\pi(S)\propto\wt(S)$.

The extremal formulation by Gorodezky and Pak \cite{GP14} is the following.
Each arc $e$ is associated with a (distinct) random variable which is present with probability $1-p_a$.
A cluster is a subset of vertices not containing $r$ so that no arc is going out.
We want all vertices to be able to reach $r$.
Thus clusters are undesirable.
However, \cref{cond:extremal} is not satisfied if we simply let the bad events be clusters.
Instead, we choose minimal clusters to be bad events.

More precisely, for each $C\subseteq V$, we associate it with a bad event $B_C$ to indicate that $C$ is a cluster
and no proper subset $C'\subsetneq C$ is one.
Each $B_C$ depends on arcs going out of $C$ for being a cluster,
as well as arcs inside $C$ for being minimal.
In other words, $\vbl(B_C)=\{u\rightarrow v\mid u\in C\}$.
There are clearly exponentially many bad events.
A description of the algorithm is given in \Cref{alg:cluster-popping}.

\begin{algorithm}
  \caption{Cluster-popping}
  \label{alg:cluster-popping}
  Let $S$ be a subset of arcs by choosing each arc $a$ with probability $1-p_a$ independently\;
  \While{there is a cluster in $(V,S)$}{
    Let $C_1,\dots,C_k$ be all minimal clusters in $(V,S)$\;
    Re-randomize all arcs in $\bigcup_{i=1}^k\vbl(B_{C_i})$ to get a new $S$\;
  }
  \KwRet{$S$}
\end{algorithm}

\Cref{cond:extremal} is met, due to the following observation \cite[Claim 3]{GJ18a}.

\begin{claim}\label{clm:cluster:sc}
  Any minimal cluster is strongly connected.
\end{claim}

Since \Cref{cond:extremal} holds,
\Cref{alg:cluster-popping} draws from the desired distribution over root-connected subgraphs in a directed graph.
It was further shown in \cite{GJ18a} that
the cluster-popping algorithm can be used to sample spanning connected subgraphs in an undirected graph,
and to approximate all-terminal network reliability in expected polynomial time.
We will first give an improved running time bound for the basic cluster-popping algorithm.

\subsection{Expected running time}

A graph $G=(V,A)$ is \emph{bi-directed} if each arc $a\in A$ has an anti-parallel twin in $A$ as well.
For an arc $a=u\rightarrow v$, let $\reverse{a} := v\rightarrow u$ be its reversal.
Let $n=\abs{V}$ and $m=\abs{A}$. 
It was shown that cluster-popping can take exponential time in general  \cite{GP14}
while in bi-directed graphs, the expected number of resampled variables is at most $\frac{p_{\maxx}}{1-p_{\maxx}}m^2n$  \cite{GJ18a},
where $p_{\maxx}=\max_{a\in A} p_a$. We will now design a $\frac{p_{\maxx}}{1-p_{\maxx}}$-preserving injective mapping
from $\Omega_1^{\vbl}$ to $\Omega_0\times V\times A$ for a connected and bi-directed graph $G$,
where $\Omega_0$ is the set of all root-connected subgraphs, and $\Omega_1$ is the set of subgraphs with a unique minimal cluster.
Applying \Cref{thm:run-time} with $\aux=V\times A$, the expected number of resampled variables is $\frac{p_{\maxx}}{1-p_{\maxx}}mn$.
Thus a factor $m$ is shaved.

Our $\frac{p_{\maxx}}{1-p_{\maxx}}$-preserving injective mapping $\varphi:\Omega_1^{\vbl}\rightarrow \Omega_0\times V\times A$ is modified from the one in~\cite{GJ18a}.
The main difference is that the domain is $\Omega_1^{\vbl}$ instead of $\Omega_1$,
and thus we need to be able to recover the random variable in the encoding.
The encoding is also more efficient.
For example, we only record $u$ instead of $u\rightarrow u'$ in~\cite{GJ18a} as explained below.

For completeness and clarity we include full details.
We assume that the bi-directed graph $G$ is connected so that $\Omega_0\neq\emptyset$,
and the root $r$ is arbitrarily chosen.
(Apparently in a bi-directed graph, weak connectivity is equivalent to strong connectivity.)

For each subgraph $S\in\Omega_1$, the rough idea is to ``repair'' $S$ so that no minimal cluster is present.
We fix in advance an arbitrary ordering of vertices and arcs.
Let $C$ be the unique minimal cluster in $S$ and $v\rightarrow v'$ be an arc so that $v\in C$,
namely $v\rightarrow v'\in \vbl(B_C)$.
Let $R$ denote the set of all vertices which can reach the root $r$ in the subgraph $S$.
Since $S\in\Omega_1$, $R\neq V$.
Let $U = V\setminus R$.
Since $G$ is root-connected, there is at least one arc in $A$ from $U$ to $R$.
Let $u \rightarrow u'$ be the first such arc, where $u\in U$ and $u'\in R$.
Let
\begin{align}\label{eqn:cluster-varphi}
  \varphi(S,v\rightarrow v') := (\sfix,u,v\rightarrow v'),
\end{align}
where $\sfix\in\Omega_0$ is defined in the same way as in \cite{GJ18a} and will be presented shortly.
In \cite{GJ18a}, the mapping is from $S$ to $(\sfix,v,u\rightarrow u')$.

Consider the subgraph $H=(U,S[U])$, where
\begin{align*}
  S[U] := \{x\rightarrow y\mid x\in U,\; y\in U,\; x\rightarrow y\in S\}.
\end{align*}
We consider the directed acyclic graph (DAG) of strongly connected components of $H$, and call it $\widehat{H}$.
(We use the decoration $\widehat{\;}$ to denote arcs, vertices, etc.\ in $\widehat{H}$.)
To be more precise, we replace each strongly connected component by a single vertex.
For a vertex $w\in U$, let $[w]$ denote the strongly connected component containing $w$.
For example, $[v]$ is the same as the minimal cluster $C$ by \cref{clm:cluster:sc}.
We may also view $[w]$ as a vertex in $\widehat{H}$ and we do not distinguish the two views.
The arcs in $\widehat{H}$ are naturally induced by $S[U]$.
Namely, for $[x]\neq[y]$, an arc $[x]\rightarrow [y]$ is present in $\widehat{H}$
if there exists $x'\in[x]$, $y'\in[y]$ such that $x'\rightarrow y' \in S$.

We claim that $\widehat{H}$ is root-connected with root $[v]$.
This is because $[v]$ must be the unique sink in $\widehat{H}$ and $\widehat{H}$ is acyclic.
If there is another sink $[w]$ where $v\not\in[w]$, then $[w]$ is a minimal cluster in $H$.
This contradicts to $S\in\Omega_1$.

Since $\widehat{H}$ is root-connected, there is at least one path from $[u]$ to $[v]$.
Let $\widehat{W}$ denote the set of vertices of $\widehat{H}$ that can be reached from $[u]$ in $\widehat{H}$ (including $[u]$),
and $W:= \{x\mid [x]\in \widehat{W}\}$.
Then $W$ is a cluster and $[u]$ is the unique source in $\widehat{H}[\widehat{W}]$.
As $\widehat{H}$ is root-connected, $[v]\in\widehat{W}$.
To define \sfix, we reverse all arcs in $S[W]$ and add the arc $u\rightarrow u'$ to eliminate the cut.
Formally, let
\begin{align*}
  \sfix := (S\setminus S[W]) \cup \{u\rightarrow u'\} \cup \{ y\rightarrow x \mid x\rightarrow y\in S[W]\}.
\end{align*}

Let $\widehat{H}_{\fix}$ be the graph obtained from $\widehat{H}$ by reversing all arcs induced by $S[W]$.
Observe that $[u]$ becomes the unique sink in $\widehat{H}_{\fix}[\widehat{W}]$ (and $[v]$ becomes the unique source).

We verify that $\sfix\in\Omega_0$.
For any $x\in R$, $x$ can still reach $r$ in $(V,\sfix)$ since the path from $x$ to $r$ in $(V,S)$ is not changed.
Since $u\rightarrow u'\in\sfix$, $u$ can reach $u'\in R$ and hence $r$.
For any $y\in W$, $y$ can reach $u$ as $[u]$ is the unique sink in $\widehat{H}_{\fix}[\widehat{W}]$.
For any $z\in U\setminus W$, $z$ can reach $v\in W$ since the path from $z$ to $v$ in $(V,S)$ is not changed.

\begin{lemma}  \label{lem:cluster-injective}
  The mapping $\varphi:\Omega_1^{\vbl}\rightarrow\Omega_0\times V\times A$
  defined in \eqref{eqn:cluster-varphi} is $\frac{p_{\maxx}}{1-p_{\maxx}}$-preserving and injective.
\end{lemma}
\begin{proof}
  It is easy to verify $\frac{p_{\maxx}}{1-p_{\maxx}}$-preservingness,
  since flipping arcs leaves its weight unchanged.
  The only move changing the weight is to add the arc $u\rightarrow u'$ in \sfix,
  which results in at most $\frac{p_{\maxx}}{1-p_{\maxx}}$ change.

  Next we verify that $\varphi$ is injective.
  To do so, we show that we can recover $S$ and $v\rightarrow v'$ given $\sfix$, $u$, and $v\rightarrow v'$.
  Clearly, it suffices to recover just $S$.

  First observe that in $\sfix$, $u'$ is the first vertex on any path from $u$ to $r$.
  Thus we can recover $u\rightarrow u'$.
  Remove it from $\sfix$.
  The set of vertices which can reach $r$ in $(V,\sfix\setminus\{u\rightarrow u'\})$ is exactly $R$ in $(V,S)$.
  Namely we can recover $U$ and $R$.
  As a consequence, we can recover all arcs in $S$ that are incident with~$R$, as these arcs are unchanged.

  What is left to do is to recover arcs in $S[U]$.
  To do so, we need to find out which arcs have been flipped.
  We claim that $\widehat{H}_{\fix}$ is acyclic.
  Suppose there is a cycle in $\widehat{H}_{\fix}$.
  Since $\widehat{H}$ is acyclic, the cycle must involve flipped arcs and thus vertices in $\widehat{W}$.
  Let $[x]\in\widehat{W}$ be the lowest one under the topological ordering of $\widehat{H}[\widehat{W}]$.
  Since $\widehat{W}$ is a cluster,
  the outgoing arc $[x]\rightarrow [y]$ along the cycle in $\widehat{H}_{\fix}$ must have been flipped,
  implying that $[y]\in\widehat{W}$ and $[y]\rightarrow [x]$ is in $\widehat{H}[\widehat{W}]$.
  This contradicts to the minimality of $[x]$.

  Since $\widehat{H}_{\fix}$ is acyclic,
  the strongly connected components of $H_{\fix}:=(U,\sfix[U])$ are identical to those of $H=(U,S[U])$.
  (Note that flipping all arcs in $S[W]$ leaves strongly connected components inside $W$ unchanged.)
  Hence contracting all strongly connected components of $H_{\fix}$ results in exactly $\widehat{H}_{\fix}$.
  All we need to recover now is the set $\widehat{W}$.
  Let $\widehat{W}'$ be the set of vertices reachable from $[v]$ in $\widehat{H}_{\fix}$.
  It is easy to see that $\widehat{W} \subseteq \widehat{W}'$ since we were flipping arcs.
  We claim that actually $\widehat{W} = \widehat{W}'$.
  For any $[x]\in\widehat{W}'$,
  there is a path from $[v]$ to $[x]$ in $\widehat{H}_{\fix}$.
  Suppose $[x]\not\in\widehat{W}$.
  Since $[v]\in\widehat{W}$, we may assume that $[y]$ is the first vertex along the path such that $[y]\rightarrow[z]$ where $[z]\not\in\widehat{W}$.
  Thus $[y]\rightarrow[z]$ has not been flipped and is present in $\widehat{H}$.
  However, this contradicts the fact that $\widehat{W}$ is a cluster in $\widehat{H}$.

  To summarize, given $\sfix$, $u$, and $v\rightarrow v'$,
  we may uniquely recover $S$ and thus $v\rightarrow v'$.
  Hence the mapping $\varphi$ is injective.
\end{proof}

Combining \Cref{lem:cluster-injective} and \Cref{thm:run-time} (with $\aux=V\times A$) implies
an upper bound of the number of random variables drawn in expectation for \Cref{alg:cluster-popping}.

\begin{theorem}\label{thm:cluster-time}
  To sample edge-weighted root-connected subgraphs,
  the expected number of random variables drawn in \Cref{alg:cluster-popping} on a connected bi-directed graph $G=(V,A)$ is
  at most $m+\frac{p_{\maxx}mn}{1-p_{\maxx}}$, where $n=\abs{V}$ and $m=\abs{A}$.
\end{theorem}

Another consequence of \Cref{lem:cluster-injective} is that we can bound the expected number of resamplings of each individual variable.
Denote by $T_a$ the number of resamplings of $a$.
Then, by \eqref{eqn:expected-i},
\begin{align*}
  \Ex T_a=\sum_{C:\;a\in\vbl(B_C)} \frac{q_C}{q_{\emptyset}},
\end{align*}
where $q_C$ is the probability that under the product distribution, there is a unique minimal cluster $C$,
and $q_{\emptyset}$ is the probability that under the product distribution, there is no cluster.
Through the $\frac{p_{\maxx}}{1-p_{\maxx}}$-preserving injective mapping $\varphi$, namely \eqref{eqn:cluster-varphi},
if we fix $a$,
it is easy to see that
\begin{align*}
  \sum_{C\in \+C_a} \frac{q_C}{q_{\emptyset}} \le \frac{p_{\maxx}n}{1-p_{\maxx}}.
\end{align*}
As the above holds for any $a$, we have that the expected number of resampling of any variable is at most $\frac{p_{\maxx}n}{1-p_{\maxx}}$.
This is an upper bound for the expected depth of the resampling table.

\subsection{An efficient implementation}

\Cref{thm:cluster-time} bounds the number of random variables drawn.
However, regarding the running time of \Cref{alg:cluster-popping},
in addition to drawing random variables,
a naive implementation of \Cref{alg:cluster-popping} may need to find clusters in every iteration of the while loop,
which may take as much as $O(m)$ time by, for example, Tarjan's algorithm~\cite{Tar72}. 
In \Cref{alg:cluster-popping-Tarjan} we give an efficient implementation, which can be viewed as a dynamic version of Tarjan's algorithm.

\Cref{alg:cluster-popping-Tarjan} is sequential,
and its correctness relies on the fact that the order of resampling events for extremal instances does not affect the final output.
See \cite[Section 4]{GJ18c} for a similar sequential (but efficient) implementation of ``bicycle-popping''.

Two key modifications in \Cref{alg:cluster-popping-Tarjan} comparing to Tarjan's algorithm are:
\begin{enumerate}
  \item in the DFS, once the root $r$ is reached, all vertices along the path are ``set'' and will not be resampled any more;
  \item the first output of Tarjan's algorithm is always a strongly connected component with no outgoing arcs. 
    Such a component (if it does not contain the root $r$) is a minimal cluster and will immediately be resampled in \Cref{alg:cluster-popping-Tarjan}.
\end{enumerate}
Let $G=(V,A)$ be the input graph.
For $v\in V$, let $\arc(v)=\{v\rightarrow w\mid v\rightarrow w\in A\}$ be the set of arcs going out from $v$.
Recall that for $C\subseteq V$, $\vbl(B_C)=\bigcup_{v\in C}\arc(v)$.

\begin{algorithm}
  \caption{Cluster-popping with Tarjan's algorithm}
  \label{alg:cluster-popping-Tarjan}
  \Input{A directed graph $G=(V,A)$ with a special root vertex $r$\;}
  \Output{$R\subseteq A$ drawn according to $\pi(\cdot)$\;}
  Let $R\subseteq A$ be obtained by drawing each arc $a\in A$ with probability $1-p_a$ independently\;
  $r.index\gets 1$, $r.root\gets 1$, $index\gets 2$\;  
  Let $S$ be a stack consisting of only $r$\;
  \While{$\exists v \in V,\;v.index = \textnormal{UNDEFINED}$}{
    Dynamic-DFS$(v)$\;
  }
  \KwRet{$R$}\;
  \Fn{\emph{Dynamic-DFS}$(v)$}{
    $v.index \gets index$\;
    $v.root \gets v.index$\;
    $index \gets index + 1$\;
    $S.push(v)$\;
    \For{\emph{each} $v\rightarrow w\in R$}{
      \uIf{$w.index = \textnormal{UNDEFINED}$}{
      Dynamic-DFS$(w)$\;
      $v.root\gets\min\{v.root, w.root\}$\;
      }
      \Else{
        $v.root\gets\min\{v.root, w.index\}$\;
        }
      }
    \If(\tcp*[f]{A minimal cluster is found}){$v.root = v.index$}{
      \Repeat{$w = v$}{
        $w\gets S.pop()$\;
        $w.index\gets \text{UNDEFINED}$,
        $w.root\gets\text{UNDEFINED}$,
        $index\gets index-1$\;
        $R\gets R\setminus\arc(w)$, draw each arc $a\in \arc(w)$ with probability $1-p_a$ independently, and add them to $R$\tcp*{Resampling step}
      }
      Dynamic-DFS$(v)$\;
    }
  }
\end{algorithm}

We first observe that in \Cref{alg:cluster-popping-Tarjan}, $index-1$ is always the number of variables that have already been indexed.
Furthermore, inside the function ``Dynamic-DFS'' of \Cref{alg:cluster-popping-Tarjan},
for any $v\in V$, if $v.root$ is defined, then $v$ can reach the vertex indexed by $v.root$.
This can be shown by a straightforward induction on the recursion depth,
and observing that resampling in any Dynamic-DFS$(v)$ will not affect arcs whose heads are indexed before $v$.

\begin{lemma}  \label{lem:Tarjan-invariant}
  At the beginning of each iteration of the while loop in \Cref{alg:cluster-popping-Tarjan},
  all vertices whose indices are defined can reach the root $r$, 
  and they will not be resampled any more.
\end{lemma}
\begin{proof}
  We do an induction on the number of while loops executed.
  The base case is trivial as only $r$ is indexed before the first iteration.
  Suppose we are to execute Dynamic-DFS$(v)$ for some $v$.
  By the induction hypothesis, for any $w$ such that $w.index<v.index$, $w$ can reach the root $r$.
  As $u.root$ is non-increasing for any $u\in V$,
  the only possibility to exit any (recursive) call of Dynamic-DFS$(u)$ is that $u.root<u.index$.
  Suppose after finishing Dynamic-DFS$(v)$, $u$ is the lowest vertex that is newly indexed but cannot reach $r$,
  and $u.root=w.index$ for some $w$.
  Then $u$ can reach $w$ and $w.index < u.index$.
  However, the latter implies that $w$ can reach $r$ by the choice of $u$, which is a contradiction.
\end{proof}

In particular, \Cref{lem:Tarjan-invariant} implies that once the algorithm halts, 
all vertices can reach the root $r$, and the output is a root-connected subgraph.

\begin{lemma}  \label{lem:Tarjan-minimal-cluster}
  Inside the function ``Dynamic-DFS$(v)$'' of \Cref{alg:cluster-popping-Tarjan},
  if $v.root=v.index$ happens, 
  then a minimal cluster is resampled.
\end{lemma}
\begin{proof}
  Suppose $v.index=v.root$.
  We want to show that all vertices indexed after $v$ form a minimal cluster.
  Let $U=\{u\mid u.index \ge v.index\}$.

  Clearly $v$ can reach all vertices in $U$.
  For any $u\in U$, if $u.root < v.index$,
  then $v.root\le u.root < v.index$, contradicting to our assumption.
  Thus $u.root\ge v.index$ for all $u\in U$.
  Let $u_0\in U$ be the lowest indexed vertex that cannot reach $v$.
  Since the recursive call of $u_0$ has exited,
  $v.index \le u_0.root < u_0.index$.
  Thus $u_0.root$ is a vertex in $U$ and can reach $v$ due to the choice of $u_0$.
  It implies that $u_0$ can reach $u_0.root$ and thus $v$, a contradiction.
  Hence, all of $U$ can reach $v$, and thus $U$ is strongly connected.

  If there is any arc in the current $R$ going out from $U$, say $u\rightarrow w$ for some $u\in U$,
  then it must be that $w.index < v.index$.
  However, this implies that $u.root \le w.index < v.index$, contradicting to $u.root\ge v.index$ for all $u\in U$.
  This implies that there is no arc going out from $U$.
  Thus, $U$ is a cluster and is strongly connected, and it must be a minimal cluster.
\end{proof}

Now we are ready to show the correctness and the efficiency of \Cref{alg:cluster-popping-Tarjan}.

\begin{theorem}  \label{thm:Tarjan}
  The output distribution of \Cref{alg:cluster-popping-Tarjan} is the same as that of \Cref{alg:cluster-popping},
  and the expected running time is $O\left(m+\frac{p_{\maxx}mn}{1-p_{\maxx}}\right)$.
\end{theorem}
\begin{proof}
  By \Cref{lem:Tarjan-invariant} and \Cref{lem:Tarjan-minimal-cluster},
  in \Cref{alg:cluster-popping-Tarjan}, only minimal clusters are resampled and the halting rule is when no minimal cluster is present.
  In other words, the resampling rules are the same for \Cref{alg:cluster-popping-Tarjan} and \Cref{alg:cluster-popping},
  except the difference in orderings.
  Given a resampling table, our claim is that the resampled variables are exactly the same for \Cref{alg:cluster-popping-Tarjan} and \Cref{alg:cluster-popping},
  which leads to an identical final state.
  The claim can be verified straightforwardly, or we can use a result of Eriksson \cite{Eri96}.
  If any two minimal clusters are present, then they must be disjoint (\Cref{cond:extremal}), 
  and thus different orders of resampling them lead to the same state for a fixed resampling table.
  This is the \emph{polygon property} in \cite{Eri96},
  which by \cite[Theorem 2.1]{Eri96} implies the \emph{strong convergence property}.
  The latter is exactly our claim.

  Regarding the running time of \Cref{alg:cluster-popping-Tarjan},
  we observe that its running time is linear in the final output and the number of resampled variables.
  The expected number of resampled variables of \Cref{alg:cluster-popping-Tarjan} is the same as that of \Cref{alg:cluster-popping} due to the claim above.
  Thus, \Cref{thm:cluster-time} implies the claimed run-time bound.
\end{proof}

We can further combine \Cref{alg:cluster-popping-Tarjan} with the coupling procedure \cite[Section 5]{GJ18a} to yield
a sampler for edge-weighted spanning connected subgraphs in an \emph{undirected} graph,
which is key to the FPRAS of all-terminal network reliability.
The coupling performs one scan over all edges.
Thus we have the following corollary of \Cref{thm:Tarjan}.

\begin{corollary}\label{cor:random-connected-subgraph}
  There is an algorithm to sample edge-weighted spanning connected subgraphs in an undirected graph $G=(V,E)$
  with expected running time $O\left(m+\frac{p_{\maxx}mn}{1-p_{\maxx}}\right)$, where $n=\abs{V}$ and $m=\abs{E}$.
\end{corollary}


\subsection{A tight example}\label{sec:tightexampleforcluster}
In this section, we present an example that the expected number of random variables drawn in \Cref{alg:cluster-popping} (and thus \Cref{alg:cluster-popping-Tarjan}) is $2m+(1-o(1))\frac{pmn}{1-p}$,
where $p_a=p$ for all $a\in A$.
Thus, the bound in \Cref{thm:cluster-time} is tight.

Our example is the bi-directed version of the ``lollipop'' graph, where a simple path $P$ of length $n_1$ is attached to a clique $K$ of size $n_2$.
A picture is drawn in \Cref{fig:lollipop}.

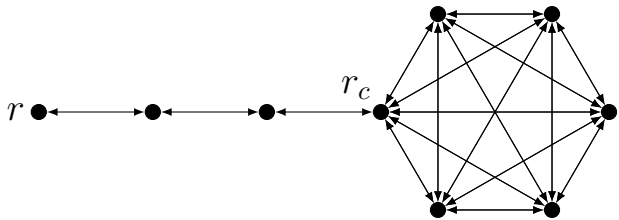
\begin{figure}[htbp]
  \centering
    \begin{tikzpicture}[scale=1, inner sep=2pt, transform shape]
      \draw (180:1.5cm) node [draw,fill,shape=circle,color=black, label=120:{\Large $r_c$}] (v0) {};
      \draw (240:1.5cm) node [draw,fill,shape=circle,color=black] (v1) {};
      \draw (300:1.5cm) node [draw,fill,shape=circle,color=black] (v2) {};
      \draw   (0:1.5cm) node [draw,fill,shape=circle,color=black] (v3) {};
      \draw  (60:1.5cm) node [draw,fill,shape=circle,color=black] (v4) {};
      \draw (120:1.5cm) node [draw,fill,shape=circle,color=black] (v5) {};

      \foreach \x in {0,1,2,3,4,5}
      \foreach \y in {0,1,2,3,4,5}
      {
        \ifthenelse{\x=\y}{}{\draw (v\x) edge [<->] (v\y);}
      }

      \draw (180:3cm) node [draw,fill,shape=circle,color=black] (u1) {};
      \draw (180:4.5cm) node [draw,fill,shape=circle,color=black] (u2) {};
      \draw (180:6cm) node [draw,fill,shape=circle,color=black,label=180:{\Large $r$}] (u3) {};

      \draw (v0) edge [<->] (u1);
      \draw (u1) edge [<->] (u2);
      \draw (u2) edge [<->] (u3);
    \end{tikzpicture}
    \caption{A bi-directed lollipop graph with $n_1=3$ and $n_2=6$.}
  \label{fig:lollipop}
\end{figure}

The main tool is still the formula \eqref{eqn:expected}.
We have constructed an injective mapping $\varphi:\Omega_1^{\vbl}\rightarrow \Omega_0\times V\times A$.
Thus, to derive a lower bound, we just need to lower bound the weighted ratio between $\varphi\left(\Omega_1^{\vbl}\right)$ and $\Omega_0\times V\times A$.
The main observation is that for most $S'\in \Omega_0$,
the tuple $(S',u, v\rightarrow v')\in \varphi\left(\Omega_1^{\vbl}\right)$
as long as $u\in P$ is not the right endpoint of $P$, and $v\rightarrow v'$ is an arc in $K$.
We may choose $n_1$ and $n_2$ so that $\abs{P}=n_1=(1-o(1))n$ and the number of arcs in $K$ is $n_2(n_2-1)=(1-o(1))m$.
The bound in \Cref{alg:cluster-popping} is tight with this choice.

For concreteness, let $n_1 = \ceil{n_2^{1+\eps}}$ for some $0<\eps<1$ and consider $n_2\rightarrow \infty$.
Then the number of vertices is $n=n_1+n_2=(1+n_2^{-\eps})n_1$ and
the number of arcs is $m=2n_1+n_2(n_2-1)=(1+o(1))n_2^2$.
Let the root $r$ be the leftmost vertex of $P$, and $r_c$ be the vertex where $P$ and $K$ intersect.
A subgraph $S'\in \Omega_0$ must contain a directed path from $r_c$ to $r$,
as well as a root-connected subgraph in $K$ with root $r_c$.
For any constant $p$, since $K$ is a clique, with high probability,
a random subgraph (by drawing each edge independently with probability $1-p$) in $K$ is strongly connected.\footnote{This fact is easy to prove. Recall that the analogue connectivity threshold for Erd\H{o}s-R\'enyi random graph is $p=\frac{\log n}{n}$.}
Let $\Omega_0'$ be the set of subgraphs which contain a directed path from $r$ to $r_c$ and strongly connected inside $K$.
Clearly the total weight of $\Omega_0'$ is $1-o(1)$ of that of $\Omega_0$.

For each $S'\in\Omega_0'$, $u\in P\setminus\set{r_c}$, and $v\rightarrow v'$ an arc in $K$,
it is straightforward to verify that the ``repairing'' procedure in \Cref{lem:cluster-injective} goes through.
We first remove the arc $u'\rightarrow u$, where $u'$ is the next vertex on $P$, and call the vertex set that cannot reach the root $U$.
Since $S'\in\Omega_0'$, if we contract strongly connected components in $S'[U]$,
it must be a directed path.
Flip all arcs in $S'[U]$ to get $S$.
The subgraph $S[U]$ must have the same collection of strongly connected components as $S'[U]$,
and the contracted graph is a directed path in the reverse direction.
Clearly there is a unique sink, namely a minimal cluster in $S[U]$, which is the clique $K$ with possibly a few vertices along the path $P$.
Hence $v\rightarrow v'$ is in the unique minimal cluster of $S$.

To summarise, $S\in\Omega_1$, $\wt(S)=\frac{p}{1-p}\cdot\wt(S')$, and $\varphi(S,v\rightarrow v')=(S',u, v\rightarrow v')$.
Thus, by~\eqref{eqn:expected}, the expected number of random variables drawn of \Cref{alg:cluster-popping} is at least
\begin{align*}
  2m+(1-o(1))\cdot\frac{p}{1-p}\cdot \abs{P\setminus\set{r_c}} \cdot n_2(n_2-1) = 2m+(1-o(1))\frac{pmn}{1-p},
\end{align*}
where the term $2m$ accounts for the initialisation.

\begin{remark}
  If we set $\eps=0$, then the running time becomes $\Omega(n^3)$.
  However, the optimal constant (measured in $n^3$) is not clear.
\end{remark}

An interesting observation is that the running time of \Cref{alg:prs} depends on the choice of $r$ in the example above,
although in the reliability approximation algorithm, $r$ can be chosen arbitrarily. (See \cite[Section 5]{GJ18a}.)
However, choosing the best $r$ does not help reducing the order of the running time.
In a ``barbell'' graph, where two cliques are joined by a path,
no matter where we choose $r$, there is a rooted induced subgraph of the same structure as the example above,
leading to the same $\Omega(n^3)$ running time when $\eps=0$.

\subsection{Faster reliability approximation}

The main application of \Cref{alg:cluster-popping} is to approximate the network reliability of a undirected graph $G=(V,E)$,
which is the probability that, assuming each edge $e$ fails with probability $p_e$ independently,
the remaining graph is still connected.
Let $\*p=\left( p_e \right)_{e\in E}$ be the vector of the failure probabilities.
Then the reliability is the following quantity:
\begin{align}\label{eqn:reliability}
  \Zrel(\mathbf{p})=\sum_{\substack{S\subseteq E\\(V,S)\text{ is connected}}}\;\prod_{e\in S}(1-p_e)\prod_{e\not\in S}p_e.
\end{align}
The approximate counting algorithm in \cite{GP14,GJ18a} takes $O(n^2/\eps^2)$ samples of spanning connected subgraphs to produce a $1\pm\eps$ approximation of $\Zrel$.
However, we can rewrite \eqref{eqn:reliability} as a partition function of the Gibbs distribution.
Thus we can take advantage of faster approximation algorithms, such as the one by Kolmogorov \cite{Kol18}.

Let $\Omega$ be a finite set, and the Gibbs distribution $\pi(\cdot)$ over $\Omega$ is one taking the following form:
\begin{align*}
  \pi_{\beta}(X) = \frac{1}{Z(\beta)}\exp(-\beta H(X)),
\end{align*}
where $\beta$ is the \emph{temperature}, $H(X)\ge0$ is the \emph{Hamiltonian},
and $Z(\beta)=\sum_{X\in\Omega}\exp(-\beta H(X))$ is the normalising factor (namely the partition function) of the Gibbs distribution.
We would like to turn the sampling algorithm into an approximation algorithm to $Z(\beta)$.
Typically, this involves calling the sampling oracle in a range of temperatures, which we denote $[\beta_{\minn},\beta_{\maxx}]$.
(This process is usually called simulated annealing.)
Let $Q:=\frac{Z(\beta_{\minn})}{Z(\beta_{\maxx})}$, $q=\log Q$, and $N=\max_{X\in\Omega} H(X)$.
The following result is due to Kolmogorov \cite[Theorem 8 and Theorem 9]{Kol18}.

\begin{proposition}  \label{prop:Kolmogorov}
  Suppose we have a sampling oracle from the distribution $\pi_{\beta}$ for any $\beta\in[\beta_{\minn},\beta_{\maxx}]$.
  There is an algorithm to approximate $Q$ within $1\pm\eps$ multiplicative errors
  using $O(q\log N/\eps^{2})$ oracle calls in average.

  Moreover, the sampling oracle $\widetilde{\mu}_\beta$ can be approximate
  as long as $\Vert\mu_\beta - \widetilde{\mu}_\beta\Vert_{TV} = O(1/(q\log N))$ where $\Vert \cdot\Vert_{TV}$ is the variation distance.
\end{proposition}

A straightforward application of \Cref{prop:Kolmogorov} to our problem requires $O(m\log n)$ samples.
This is because annealing will be done on all edges.
Instead, we will choose a spanning tree, and perform annealing only on its edges, whose cardinality is $n-1$.
This approach uses only $O(n\log n)$ samples,
but it requires the following slight generalisation of \Cref{prop:Kolmogorov}.

Let $\rho_{\beta}(\cdot)$ be the following distribution over a finite set $\Omega$:
\begin{align}\label{eqn:rho}
  \rho_{\beta}(X) = \frac{1}{Z(\beta)}\exp(-\beta H(X))\cdot F(X),
\end{align}
where $F:\Omega\rightarrow\mathbb{R}^+$ is a non-negative function
and, with a little abuse of notation, $Z(\beta)=\sum_{X\in\Omega}\exp(-\beta H(X))\cdot F(X)$ is the normalising factor.
Still, let $Q:=\frac{Z(\beta_{\minn})}{Z(\beta_{\maxx})}$, $q=\log Q$, and $N=\max_{X\in\Omega} H(X)$.

\begin{lemma}  \label{lem:Kolmogorov-general}
  Suppose we have a sampling oracle from the distribution $\rho_{\beta}$ defined in \eqref{eqn:rho} for any $\beta\in[\beta_{\minn},\beta_{\maxx}]$.
  There is an algorithm to approximate $Q$ within $1\pm\eps$ multiplicative errors using $O(q\log N/\eps^{-2})$ oracle calls.
\end{lemma}
\begin{proof}
  We claim that we can straightforwardly apply the algorithm in \Cref{prop:Kolmogorov} to get an approximation to $Q$ for $\rho_{\beta}$.

  To see this, let $\ell\ge 0$ be an integer.
  Let $\Omega'$ be the multi-set of $\Omega$ where each $X$ is duplicated $\floor{\ell F(X)}$ times.
  To avoid multi-set, we can simply give each duplicated $X$ an index to make $\Omega'$ an ordinary set.
  Consider the following Gibbs distribution over $\Omega'$:
  \begin{align*}
    \pi_{\beta}(X) = \frac{1}{Z'(\beta)}\exp(-\beta H(X)),
  \end{align*}
  where $Z'(\beta)=\sum_{X\in\Omega'}\exp(-\beta H(X))$.
  Let
  \begin{align*}
    \delta_{\ell}:=\min_{X\in\Omega}\left(1-\frac{\floor{\ell F(X)}}{\ell F(X)}\right).
  \end{align*}
  We have that
  \begin{align*}
    1-\delta_{\ell} \le \frac{Z'(\beta)}{\ell Z(\beta)} \le 1.
  \end{align*}
  Clearly, as $\ell\rightarrow \infty$, $\delta_{\ell}\rightarrow 0$.
  We choose $\ell$ large enough
  so that $\delta_{\ell}=O(1/q\log N)$ is within the threshold in \Cref{prop:Kolmogorov} for approximate sampling oracle.
  Thus, the output of directly applying the algorithm in \Cref{prop:Kolmogorov} with sampling oracle $\rho_{\beta}$ also yields an $\eps$-approximation to $Q':=\frac{Z'(\beta_{\minn})}{Z'(\beta_{\maxx})}$.
  (Note that we do not really run the algorithm on $\Omega'$.)
  Since $\frac{Q'}{Q}\rightarrow 1$ as $\ell\rightarrow \infty$,
  the output of directly applying \Cref{prop:Kolmogorov} is in fact an $\eps$-approximation to $Q$.
\end{proof}

Suppose we want to evaluate $\Zrel(\mathbf{p})$ for a connected undirected graph $G=(V,E)$. Since $G$ is connected, $m\ge n-1$, where $m=\abs{E}$ and $n=\abs{V}$.
Fix an arbitrary spanning tree $T$ of $G$ in advance. Let
\begin{align*}
  F(S):=\prod_{e \in T\setminus S} \frac{p_e}{1-p_e}\prod_{e \in S\setminus T}(1-p_e) \prod_{e \in E\setminus (T \cup S)} p_{e}.
\end{align*}
Let $\Omega_0$ be the set of all spanning connected subgraphs of $G$,
and $\rho_{\beta}(S)$ be the following distribution over~$\Omega_0$:
\begin{align*}
  \rho_{\beta}(S)=\frac{1}{Z(\beta)}\exp(-\beta H(S))\cdot F(S),
\end{align*}
where $\beta\ge 0$ is the temperature, $H(S):=\abs{T\setminus S}=n-1-\abs{T\cap S}$ is the Hamiltonian,
and the normalising factor $Z(\beta)=\sum_{S\in\Omega_0}\exp(-\beta H(S))\cdot F(S).$
For any $\beta\ge 0$, let $0<p'_{e}\le p_e$ be the probability such that $\frac{p_e \exp(-\beta)}{1-p_e} = \frac{p'_e}{1-p'_e}$.
We have that for any $S\in\Omega_0$,
\begin{align*}
  \rho_{\beta}(S) &=\frac{1}{Z(\beta)}\prod_{e \in T\setminus S} \frac{p_e \exp(-\beta)}{1-p_e} \prod_{e \in S\setminus T}(1-p_e) \prod_{e \in E\setminus (T \cup S)} p_{e}
  \\ &= \frac{1}{Z(\beta)}\prod_{e \in T\setminus S} \frac{p'_e }{1-p'_e}\prod_{e \in S\setminus T}(1-p_e) \prod_{e \in E\setminus (T \cup S)} p_{e}.
\end{align*}

To draw a sample from $\rho_{\beta}$, we use \Cref{cor:random-connected-subgraph},
for a vector $\*p$ such that every $e\in T$ fails with probability $p'_{e}$,
and every other $e\not\in T$ fails with probability $p_e$.
To see that this recovers the distribution $\rho_{\beta}$, notice that the weight $\wt(S)$ assigned to each subgraph $S\in\Omega_0$ is
\begin{align*}
  \wt(S) & = \prod_{e \in T\cap S}(1-p'_e)\prod_{e \in T\setminus S} p'_{e} \prod_{e \in S\setminus T}(1-p_e) \prod_{e \in E\setminus (T \cup S)} p_{e}\\
  & =  \rho_{\beta}(S)\cdot Z(\beta)\prod_{e \in T}(1-p'_e)\propto \rho_{\beta}(S).
\end{align*}

Let $\beta_{\minn}=0$ and $\beta_{\maxx}=\infty$.
Indeed, $\beta=\infty$ corresponds to $p'_e=0$.
Hence, $\rho_{\infty}(S)\neq 0$ if and only if $T\subseteq S$.
This condition implies that $S\in\Omega_0$, and
\begin{align*}
  \exp(-\infty\cdot H(S))\cdot F(S)=
  \begin{cases}
    \prod_{e\in S\setminus T}(1-p_e)\prod_{e\in E\setminus (T \cup S)}p_e & \text{ if $T\subseteq S$;}\\
    0 & \text{ otherwise.}
  \end{cases}
\end{align*}
Thus, $Z(\infty) = 1$.

On the other hand, $Z(0) = \frac{\Zrel(\*p)}{\prod_{e \in T}(1-p'_e)}$.
Then
\begin{align*}
  Q=\frac{Z(0)}{Z(\infty)} = \frac{\Zrel(\*p)}{\prod_{e \in T}(1-p'_e)} \le \prod_{e \in T}(1-p_e)^{-1},
\end{align*}
and $q=\log Q\le (n-1)\log\frac{1}{1-p_\maxx}$.
Clearly, multiplicatively approximating $Q$ is the same as approximating $\Zrel(\*p)$.
Moreover, $\max_{S\in\Omega_0}H(S)\le n-1$.
Thus, applying \Cref{lem:Kolmogorov-general},
we have an approximation algorithm for $Z(\beta)$ with $O(q\log N/\eps^2)=O(n\log n \log \frac{1}{1-p_\maxx}/\eps^2)$ oracle calls in expectation.
Combining with \Cref{cor:random-connected-subgraph}, we have the following theorem.

\begin{theorem}  \label{thm:reliability-faster}
  There is fully polynomial-time randomised approximation scheme for $\Zrel(\mathbf{p})$,
  which runs in time $O\left(\frac{mn^2\log n}{\eps^{2}(1-p_{\maxx})}\cdot\log \frac{1}{1-p_{\maxx}}\right)$ for an $(1\pm\eps)$-approximation.
\end{theorem}

\section{Cycle-popping}

Cycle-popping \cite{Wilson96,PW98} is a very simple algorithm to sample spanning trees in a connected undirected graph $G=(V,E)$.
This algorithm actually generates rooted trees, so we will pick a special root vertex $r$.
Of course, there is no real difference between rooted and unrooted spanning trees
as the root can be chosen arbitrarily in any given spanning tree.

We consider a slightly more general setting, by giving each edge $e$ a weight $w_e>0$.
For each vertex $v\in V$ other than $r$, we associate a random arc with $v$,
which points to a neighbour of $v$ so that $u$ is chosen with probability proportional to $w_{(u,v)}$.
We denote such an assignment by $\sigma$, and $\sigma(v)$ is the neighbour of $v$ that is pointed at.
Any such $\sigma$ induces a directed graph with $n-1$ arcs, where $n=\abs{V}$.
The weight of $\sigma$ is $\wt(\sigma):=\prod_{v\in V,\;v\neq r}w_{(v,\sigma(v))}$,
and the target distribution $\pi(\cdot)$ is $\pi(\sigma)\propto\wt(\sigma)$
with support on the set of $\sigma$ that induces a spanning tree.

Since we want trees in the end, cycles will be our bad events.
For each cycle $C$, we associate with it a bad event $B_C$ which indicates the presence of $C$.
The set $\vbl(B_C)$ consists of random arcs associated with vertices along $C$.
It is clear that if there is no cycle, the graph must be a spanning tree.
The number of bad events can be exponentially large, since there can be exponentially many cycles in $G$.
A description is given in \Cref{alg:cycle-popping}.

\begin{algorithm}
  \caption{Cycle-popping}
  \label{alg:cycle-popping}
  For each vertex $v\neq r$, let $\sigma(v)$ be a neighbour of $v$ with probability proportional to $w_e$ independently\;
  \While{there is at least one cycle under $\sigma$}{
    Find all cycles $C_1,\dots,C_k$ under $\sigma$\;
    Re-randomize $\bigcup_{i=1}^k\vbl(B_{C_i})$\;
  }
  \KwRet{the subgraph induced by the final $\sigma$}
\end{algorithm}

Condition \ref{cond:extremal} is easy to verify.
Two cycles are dependent if they share at least one vertex.
Suppose a cycle $C$ is present, and $C'\neq C$ is another cycle that shares at least one vertex with $C$.
If $C'$ is also present,
then we may start from any vertex $v\in C\cap C'$, and then follow the arrows $v\to v'$.
Since both $C$ and $C'$ are present, it must be that $v'\in C\cap C'$ as well.
Continuing this argument we will go back to $v$ eventually.
Thus $C=C'$. Contradiction!

\subsection{Expected running time}

Next we turn to the running time of the cycle-popping algorithm,
and define our injective mapping $\varphi:\Omega_1^{\vbl}\rightarrow \Omega_0\times V\times A$,
where $A$ is the set of all ordered pairs from $E$, namely, $A=\{u\rightarrow v,v\rightarrow u\mid (u,v)\in E\}$.
Hence $\abs{A}=2\abs{E}$.

We fix in advance an arbitrary ordering of vertices and edges.
Let $\sigma\in\Omega_{C}\subseteq\Omega_1$ be an assignment of random arcs so that there is a unique cycle $C$.
Let $u$ be a vertex on the cycle $C$ and suppose $\sigma(u)=u'$.
It is easy to see that there are two components in the subgraph induced by $\sigma$:
a directed tree with root $r$, and the directed cycle with a number of directed subtrees rooted on the cycle.
Since $G$ is connected, there must be an edge joining the two components.
Let this edge be $(v_0,v_1)$, where $v_0$ is in the tree component and $v_1$ is in the unicyclic component.
Starting from $v_1$ and following arcs under $\sigma$, we will eventually reach the cycle and arrive at $u$.
Let vertices along this path be $v_2,v_3,\dots,v_{\ell}=u$. (It is possible that $\ell=1$ or $u'$ is along the path.)

To ``fix'' $\sigma$, we reassign the arrow out of $v_i$ from $v_{i+1}$ to $v_{i-1}$ for all $1\le i\le \ell$
(in case of $v_{\ell}=u$, it is rerouted from $u'$ to $v_{\ell-1}$),
and call the resulting assignment $\sigmafix$.
Namely, $\sigmafix(v_i)=v_{i-1}$ for all $i\in [\ell]$, and $\sigmafix(v)=\sigma(v)$ otherwise.
It is easy to verify that $\sigmafix\in\Omega_0$.

Now we are ready to define the injective mapping $\varphi$.
For $\sigma\in\Omega_{C}\subseteq\Omega_1$ and $u\in C$, let
\begin{align}\label{eqn:cycle-varphi}
  \varphi(\sigma,u):=(\sigmafix,v_0, u\rightarrow u'),
\end{align}
where $\sigmafix$ and $v_0$ are defined above, and $u'=\sigma(u)$.

Let $r_{\maxx}:=\max_{e,e'\in E}\frac{w_e}{w_{e'}}$.

\begin{lemma}  \label{lem:cycle-injective}
  The mapping $\varphi:\Omega_1^{\vbl}\rightarrow \Omega_0\times V\times A$ defined in \eqref{eqn:cycle-varphi}
  is $r_{\maxx}$-preserving and injective.
\end{lemma}

\begin{proof}
  To verify that $\varphi$ is injective, we just need to recover $\sigma$ given $\sigmafix$, $v_0$, and $u\rightarrow u'$.
  Since $\sigmafix$ is a spanning tree, there is a unique path between $v_0$ and $u$ under $\sigmafix$.
  This recovers $v_1,v_2,\dots,v_{\ell}$,
  and $\sigma(v)=\sigmafix(v)$ for all other vertices since their assignments are unchanged.
  Moreover, to recover $\sigma(v_i)$ for $i\in[\ell]$ is also easy --- $\sigma(v_i)=v_{i+1}$ for $i\in[\ell-1]$ and $\sigma(u)=u'$.

  To verify that $\varphi$ is $r_{\maxx}$-preserving,
  just notice that directions do not affect weights,
  and then the only difference between $\sigma$ and $\sigmafix$ is the removal of $(u,u')$ and the inclusion of $(v_1,v_0)$.
  Thus the change in weights is at most $r_{\maxx}$.
\end{proof}

Combining \Cref{lem:cycle-injective} and \Cref{thm:run-time} (with $\aux=V\times A$) implies
a run-time upper bound of the cycle-popping algorithm.

\begin{theorem}\label{thm:cycle-time}
  To sample edge-weighted spanning trees,
  the expected number of random variables drawn in \Cref{alg:cycle-popping} on a connected undirected graph $G=(V,E)$ is
  at most $(n-1)+2r_{\maxx}mn$, where $n=\abs{V}$, $m=\abs{E}$, and $r_{\maxx}:=\max_{e,e'\in E}\frac{w_e}{w_{e'}}$.
\end{theorem}

\begin{remark}
  It is very tempting to use $\aux=V\times E$ instead of $V\times A$ in the proof above,
  which would yield an improvement by a factor $2$.
  Unfortunately that choice does not work.
  The reason is that,
  if we use an unordered pair $(u,u')$, then given $\sigmafix$ and $v_0$,
  it is not always possible to distinguish $u$ and $u'$.
  To see this,
  consider an assignment $\sigma\in\Omega_{C}$,
  and another assignment $\sigma'\in\Omega_{C}$ which is the same as $\sigma$ except that the orientation on $C$ is reversed.
  It is easy to check that the ``fixed'' versions of $(\sigma,u)$ and $(\sigma',u')$ are the same if we do not record the direction of $(u,u')$.
  We will see next that the factor $2$ in fact is unavoidable.
\end{remark}

\subsection{A tight example}

We also give a matching lower bound to complement \Cref{thm:cycle-time}.
Recall \Cref{sec:tightexampleforcluster}.
The general strategy is to construct an example
where $\varphi(\Omega_1^{\vbl})$ constitutes most of $\Omega_0\times V\times A$,
and then invoke \eqref{eqn:expected}.

We use the undirected version of the same ``lollipop'' graph as in \Cref{sec:tightexampleforcluster}.
(Recall \Cref{fig:lollipop}.)
Namely, a clique $K$ of size $n_2$ joined with a path $P$ of length $n_1$,
where $n_1 = \ceil{n_2^{1+\eps}}$ for some $0<\eps<1$.
Consider $n_2\rightarrow \infty$.
The number of vertices is $n=n_1+n_2=(1+n_2^{-\eps})n_1$ and the number of edges is $m=n_1+\frac{n_2(n_2-1)}{2}=(1/2+o(1))n_2^2$.
Let the root $r$ be the leftmost vertex of $P$, and $r_c$ be the vertex where $P$ and $K$ intersect.
Moreover, we put weight $w_c$ on all edges in $K$,
and $w_p\le w_c$ on all edges in $P$.

For a tuple $(\sigma',v_0,u\rightarrow u')$ to belong to $\varphi(\Omega_1^{\vbl})$,
the main constraint is that $v_0$ should be an ancestor of $u$ in the spanning tree induced by $\sigma'$.
In this example, any vertex $v\in P$ is an ancestor of any other vertex $u\in K$ in an arbitrary spanning tree.
Thus, for any $\sigma'\in\Omega_0$, $v_0\in P\setminus\set{r_c}$, and $u\rightarrow u'$ where $u,u'\in K$,
we can apply the ``repairing'' procedure as given in \Cref{lem:cycle-injective} to get $\sigma\in\Omega_1$
so that $\varphi(\sigma,u)= (\sigma',v_0,u\rightarrow u')$.
This is easy to verify, by finding the unique path between $v_0$ and $u$, and then reassign $\sigma$ along the path.
Since we remove one edge on the path and include one edge in the clique,
$\wt(\sigma)=\frac{w_c}{w_p}\cdot\wt(\sigma')=r_{\maxx}\wt(\sigma')$.
Thus, by~\eqref{eqn:expected}, the expected number of random variables drawn of \Cref{alg:cycle-popping} is at least
\begin{align*}
  n-1+r_{\maxx}\cdot \abs{P\setminus\set{r_c}} \cdot n_2(n_2-1) = n-1+(2-o(1))r_{\maxx}mn,
\end{align*}
where the term $n-1$ accounts for the initialisation.

Similarly to \Cref{sec:tightexampleforcluster},
the choice of $r$ affects the running time of \Cref{alg:cycle-popping}.
Indeed, in Wilson's algorithm \cite{Wilson96,PW98},
the root is chosen randomly according to the stationary distribution of the random walk,
and with that choice the running time is $O(n^2)$ in a lollipop graph (with $\eps=0$, namely $n_1\asymp n_2$).
However, also similarly to \Cref{sec:tightexampleforcluster},
the running time in a ``barbell'' graph is still $\Omega(n^3)$.
Once again, the optimal constant measured in $n^3$ is still not clear.

\section{Sink-popping}

In this section, we describe and analyse the sink-popping algorithm by Cohn, Pemantle, and Propp \cite{CPP02}.
The goal is to sample a sink-free orientation in an undirected graph.
This problem was introduced by Bubley and Dyer \cite{BD97a}
as an early showcase of the power of path coupling for Markov chains \cite{BD97b}.
This problem was also reintroduced more recently to show lower bounds for distributed Lov\'asz local lemma algorithms \cite{BFHKLRSU16},
where the goal is to find, instead of to sample, a sink-free orientation.

The formulation is as follows.
In an undirected graph $G=(V,E)$, we associate a random variable $a_e$ to indicate the orientation for each edge $e\in E$,
and associate a bad event $B_v$ for each $v\in V$ to indicate that $v$ is a sink.
Then $\vbl(B_v)=\{a_{e}\mid\text{$v$ is an endpoint of $e$}\}$,
and $\abs{\vbl(B_v)}=d_v$ where $d_v$ is the degree of $v$.
\Cref{cond:extremal} is easy to verify --- if a vertex $v$ is a sink, then none of its neighbours can be a sink.
A description of the algorithm is given in \Cref{alg:sink-popping}.

\begin{algorithm}
  \caption{Sink-popping}
  \label{alg:sink-popping}
  Orient each edge independently and uniformly at random\;
  \While{there is at least one sink}{
    Re-orient all edges that are adjacent to sinks uniformly at random\;
  }
  \KwRet{the final orientation}
\end{algorithm}

As usual, let $\Omega_v$ be the set of orientations where $v\in V$ is the unique sink,
$\Omega_1=\bigcup_{v\in V}\Omega_v$, and $\Omega_0$ be the set of all sink-free orientations.
The set $\Omega_1^{\vbl}$ is also defined as usual.
The general strategy is once again to ``repair'' orientations in $\Omega_v$.
The first step is to associate a path to each $v'\in V$ such that it can be flipped without creating new sinks,
and $v'$ is guaranteed not a sink.
For each $(v,v')\in E$, we then flip $v\leftarrow v'$, and flip the path if $v'$ is a sink now.
However, there are a few cases where we cannot recover the original orientations
if we simply record $v'$ and the other endpoint of the path.
For example, if $v$ is along the path, then $v\leftarrow v'$ is flipped twice, and there is no hope to find out $v$.
There are ways to fix these ``special'' cases, and it is relatively straightforward to design the mapping if we are happy to hardcode each special case.
However, to achieve a tight bound, we will do a more complicated mapping so that the special cases can be detected given the image.

\newcommand{\depth}{\mathrm{depth}}

A simple observation is that a graph has a sink-free orientation if and only if no connected component of it is a tree.
Moreover, the expected running time of \Cref{alg:sink-popping} on a graph with more than one component
is simply the sum of the expected running time of each component.
We will assume that $G$ is connected and not a tree, and denote $n=\abs{V}$ and $m=\abs{E}$.
Since $G$ is not a tree, there must be a cycle $C$ in $G$.
Contract the cycle $C$, and pick an arbitrary spanning tree in the resulting graph.
Denote by $R\subseteq E$ this spanning tree combined with the cycle $C$.
Thus $\abs{R}=\abs{V}=n$ and $(V,R)$ is unicyclic,
namely $(V,R)$ is composed of $C$ attached with a number of trees.
Define $\depth(v)$ by the distance from $v$ to $C$ in $(V,R)$, and $\depth(v)=0$ if $v\in C$.

Fix an arbitrary ordering of all vertices and edges of $G$.
Let $\Gamma(v)=\{v'\mid (v,v')\in E\}$ be the neighbourhood of $v$ in $G$.
An equivalent way of writing $\Omega_1^{\vbl}$ is $\{(\sigma,v')\mid\sigma\in\Omega_v,\;v'\in\Gamma(v)\}$.
We say a vertex is \emph{good} (under an orientation $\sigma$) if it has at least two outgoing arcs in $(V,R)$.

\begin{lemma}  \label{lem:sink-degree}
  Let $\sigma\in\Omega_v$ be an orientation with the unique sink $v$.
  Restricted to the unicyclic subgraph $(V,R)$,
  there must be a good vertex $u$ such that $\depth(u)<\max\{1,\depth(v)\}$.
  In particular, if $\depth(u)=\depth(v)=0$, $u$ can be chosen so that it has two outgoing arcs in $C$.
\end{lemma}
\begin{proof}
  First we prune all vertices below $v$ and other trees not containing $v$ attached to $C$.
  Call this remaining graph $H_v$, which is still unicyclic and has as many edges as its vertices.
  (If $v\in C$ then $H_v=C$.)
  Clearly $v$ in $H_v$ is still a sink under $\sigma$.
  Namely the out-degree of $v$ is $0$ in $H_v$.
  The total out-degree is same as the number of vertices in $H_v$.
  Thus, there must be a vertex $u$ with out-degree at least~$2$.
  By construction $u$ satisfies the other requirements as well.
\end{proof}

We note that $\depth(u)<\max\{1,\depth(v)\}$ implies that $\depth(u)<\depth(v)$ if $v\not\in C$,
and $\depth(u)=\depth(v)=0$ otherwise.
For $\sigma\in\Omega_v$ and $v'\in\Gamma(v)$,
choose a vertex $u$ as follows.
\begin{itemize}
  \item If $(v,v')\in R$, then let $u$ be the good vertex in \Cref{lem:sink-degree} that is closest to $v'$.
  \item If $(v,v')\not\in R$, then consider whether $v'$ is a sink in the pruned subgraph $H_{v'}$ of $(V,R)$ defined as above.
  \begin{itemize}
    \item If $v'$ is a sink in $H_{v'}$, then we apply \Cref{lem:sink-degree} to $H_{v'}$ and sink $v'$.
      Choose the closest good vertex to $v'$.
      Note that in this case, if $u$ is on $C$, then it must have at least one outgoing arc in $C$ and not along the path between $u$ and $v'$.
    \item If $v'$ is not, then we choose $u=v'$.
  \end{itemize}
\end{itemize}
All ties are broken according to the ordering chosen a priori.
Observe that if $u\neq v'$, then $\depth(u) < \max\{1,\depth(v')\}$.

We ``repair'' $\sigma$ by flipping a path between $u$ and $v'$ in $(V,R)$.
If $u$ and $v'$ are both on $C$, then we choose the one that does not contain $v$.
(If neither contains $v$, then we pick the shortest one.
Further ties are broken according to the ordering chosen a priori.)
Otherwise we simply choose the shortest path between $u$ and $v'$ in $(V,R)$.
(Again, ties are broken according to the ordering chosen a priori.)
Denote this path by $v_{\ell}=u,v_{\ell-1},\dots,v_1,v_0=v'$ where $\ell\ge 0$.
After flipping the path, we further flip the edge $(v,v')$ and denote the resulting orientation by $\sigmafix$.

We claim that $\sigmafix\in\Omega_0$.
If $\ell=0$, then $v'$ has at least two outgoing arcs and only $v\leftarrow v'$ is flipped. The claim holds.
Otherwise $\ell\ge 1$ and none of $u=v_{\ell},\dots,v_2$ can be a sink under $\sigmafix$.
For $u$, this is because it is good under $\sigma$, and only one of its adjacent edges is flipped.
For $v_{\ell-1},\dots,v_2$, they cannot have two outgoing arcs under $\sigma$ since $u$ is the closest good vertex to $v'$.
Thus after flipping, at least one of their adjacent edges along the path is still outgoing.
For $v_1$, $v_0$ and $v$, there are two cases:
\begin{enumerate}
  \item if $v_1\neq v$, then $v_1$ cannot be the sink either due to the same reasoning above.
    Moreover, since $u\neq v'$, it must be $v_1\rightarrow v'\rightarrow v$ under $\sigma$,
    and $v_1\leftarrow v'\leftarrow v$ under $\sigmafix$.
    Hence $v'$ and $v$ are not sink either;
  \item otherwise $v_1=v$.
    In this case $v\leftarrow v'$ is flipped twice.
    It must be that $\ell\ge 2$, and thus $v$ is not a sink as $v_2\leftarrow v=v_1$ under $\sigmafix$.
    No orientation of an adjacent edge of $v'$ is changed, so it is still not a sink.
\end{enumerate}
All other vertices are unchanged between $\sigma$ and $\sigmafix$.
Thus, $\sigmafix\in\Omega_0$.

Observe that $v\leftarrow v'$ is flipped twice only if $\depth(v')=\depth(v)+1$ and $(v,v')\in R$.
We say $(\sigma,v')$ is \emph{special} if $(v,v')\in R$ and,
\begin{enumerate}
  \item $\depth(v)=\depth(v')=\depth(u)=0$, or,
  \item $\depth(v')=\depth(v)+1$.
\end{enumerate}
Define the mapping $\varphi:\Omega_1^{\vbl}\rightarrow\Omega_0\times \aux$:
\begin{align}\label{eqn:sink-varphi}
  \varphi(\sigma,v')=
  \begin{cases}
    (\sigmafix,u,v') & \text{if $(\sigma,v')$ is special};\\
    (\sigmafix,v,u) & \text{otherwise},
  \end{cases}
\end{align}
where $\sigma\in\Omega_v$, $v'\in\Gamma(v)$, and $\aux$ is the set of all possible pairs of vertices appearing in the definition above.

We make a simple observation first.

\begin{lemma}  \label{lem:sink-u}
  Let $v,u,\sigmafix$ be given as above.
  If $\depth(v)=\depth(u)=0$,
  then $u$ must have at least one outgoing arc in $C$ under $\sigmafix$.
\end{lemma}
\begin{proof}
  By our choice of $u$, it has either two outgoing arcs in $C$ under $\sigma$ (when $u$ is chosen directly using \Cref{lem:sink-degree}),
  or one outgoing arc in $C$ and one pointing towards $v'$ (when $u$ is chosen by applying \Cref{lem:sink-degree} to $H_{v'}$).
  After repairing, $\sigmafix$ only flips one of the adjacent edge of $u$,
  leaving the other outgoing arc unchanged.
  In particular, in the latter case, the one flipped is not in~$C$.
\end{proof}

The main technical lemma is the following.

\begin{lemma}  \label{lem:sink-injective}
  The mapping $\varphi$ is $1$-preserving and injective.
\end{lemma}
\begin{proof}
  Since there is no weight involved, $\varphi$ is $1$-preserving.
  To verify the injectivity,
  we just need to recover $(\sigma,v')$ from $(\sigmafix,v_1,v_2)$.
  We need to figure out whether $(\sigma,v')$ is special first.
  There are a few cases:
  \begin{itemize}
    \item First we check if $v_1$ is a sink under $\sigmafix$ in $(V,R)$.
      This happens if and only if $(v,v')\not\in R$, $v_1=v$, and $(\sigma,v')$ is not special.
    \item We may now assume $(v,v')\in R$ and $v_1$ is not a sink.
      If $v_1=v_2$, then it must be a special case.
      Otherwise no matter whether $\depth(v')=\depth(v)\pm1$ or $\depth(v')=\depth(v)=0$,
      $\depth(u)\le \depth(v)$ and $\depth(u)\le\depth(v')$.
      Thus, if $\depth(v_1)<\depth(v_2)$ or $\depth(v_1)>\depth(v_2)$,
      the one with smaller depth is $u$, and we can tell whether it is special or not.

      The remaining case is that $\depth(v_1)=\depth(v_2)=0$.
      This case must be special, since if not, $v_1=v$ and $\depth(v)=0$.
      Since $(v,v')\in R$, either $\depth(v')=1$ or $\depth(v')=0$.
      Both cases are special.
      Contradiction.
  \end{itemize}
  In summary, we can distinguish the special case and the other one.

  If $(\sigma,v')$ is not special,
  then $v$ has a unique outgoing arc under $\sigmafix$, which is $v\rightarrow v'$.
  We recover $v'$ and thus the path between $u$ and $v'$, and it is easy to figure out the rest.

  If $(\sigma,v')$ is special.
  We handle the two special cases differently:
  \begin{enumerate}
    \item if $\depth(v')=\depth(u)=0$, then the path between $v'$ and $u$ must goes from $v'$ to $u$ under $\sigmafix$ along $C$.
      Moreover, $v$ must be the unique vertex going toward $v'$ under $\sigmafix$,
      since we choose the path between $v'$ and $u$ not passing $v$.
    \item if $\depth(v')=\depth(v)+1$, then the path between $v'$ and $u$ is the shortest path,
      and $v$ must be the ancestor of $v'$ in $R$ (namely the first vertex on the path from $v'$ to $u$).
  \end{enumerate}
  In both cases, we can recover the path between $v'$ and $u$ as well as the original sink $v$.
  Given these information, recovering $\sigma$ is straightforward.
\end{proof}

\Cref{lem:sink-injective} directly gives an $n^2$ upper bound for $\frac{\abs{\Omega_1^{\vbl}}}{\abs{\Omega_0}}$.
We can improve it to $n(n-1)$.
For any fixed $\sigmafix$, we want to bound the number of pairs of vertices that can be the possible output of $\varphi$ along with $\sigmafix$.
If $v$ is fixed in the non-special case of \eqref{eqn:sink-varphi}, then $u\neq v$.
In the special case, if $u$ is fixed, then in the first special case,
$v'$ cannot be the vertex $u$ points to along $C$ under $\sigmafix$ (recall \Cref{lem:sink-u});
whereas in the second special case, $v'$ cannot be $u$.
Thus \Cref{lem:sink-injective} implies that $\frac{\abs{\Omega_1^{\vbl}}}{\abs{\Omega_0}}\le n(n-1)$.
Then \eqref{eqn:expected} yields the following theorem.

\begin{theorem}  \label{thm:sink-time}
  Let $G=(V,E)$ be a connected graph that is not a tree.
  The expected number of random variables drawn in \Cref{alg:sink-popping} on $G$ to sample a sink-free orientation
  is at most $m+n(n-1)$,
  where $n=\abs{V}$ and $m=\abs{E}$.
\end{theorem}

It is easy to see that \Cref{thm:sink-time} is tight,
since on a cycle of length $n$ the upper bound is achieved.

\section{Concluding remarks}

Perhaps the most interesting open problem is whether there are faster algorithms to sample root-connected subgraphs or sink-free orientations.
For spanning trees, Kelner and M\k{a}dry \cite{KM09} has shown that the random walk based approach can be accelerated to $O(m\sqrt{n})$,
which has kindled a sequence of improvements \cite{MST15,DKPRS17,DPPR17}.
This line of research culminates in the almost-linear time algorithm by Schild \cite{Sch18}.
It would be interesting to see whether these (such as graph sparsification and fast linear system solver) or other ideas
can be used to accelerate cluster-popping and sink-popping as well.

\section*{Acknowledgements}

Part of the work was done while HG was visiting the Institute of Theoretical Computer Science, Shanghai University of Finance and Economics,
and he would like to thank their hospitality.
HG would also like to thank Mark Jerrum for helpful discussion.
KH would like to thank Xiaoming Sun for helpful discussion.

\bibliographystyle{alpha}
\bibliography{PRS}

\end{document}